\documentclass[aps,prl,twocolumn,superscriptaddress]{revtex4-2}

\usepackage{amsmath,amssymb,amsthm}
\usepackage{graphicx}
\usepackage{xcolor}
\usepackage{hyperref}
\usepackage{braket}
\usepackage{mathtools}
\usepackage{tikz}
\usepackage{pgfplots}
\pgfplotsset{compat=1.18}
\usetikzlibrary{shapes,shapes.symbols,arrows.meta,positioning,calc}

\newtheorem{theorem}{Theorem}
\newtheorem{lemma}{Lemma}
\newtheorem{corollary}{Corollary}
\newtheorem{remark}{Remark}

\newcommand{\symp}{\Omega}
\newcommand{\bvec}[1]{\boldsymbol{#1}}

\begin{document}


\title{Symplectic Barnes-Wall GKP Codes:\\Deterministic $O(N \log^2 N)$ Decoding and Logarithmic Rate Scaling}

\author{Shanxiang Lyu}
\email{lsx07@jnu.edu.cn}
\affiliation{College of Cyber Security, Jinan University, Guangzhou, China}

\date{\today}

\begin{abstract}
We construct an explicit symplectic realization of the Barnes-Wall lattice that yields a family of multimode Gottesman-Kitaev-Preskill (GKP) codes with encoding rate $R = \frac{1}{2}\log_2 N$ and a deterministic $O(N\log^2 N)$ bounded-distance decoder. The recursive generator $G_{m+1} = \bigl(\begin{smallmatrix} G_m & 0 \\ G_m & R_m G_m \end{smallmatrix}\bigr)$ with $R_m = I + \Omega$ simultaneously guarantees symplectic integrality for valid quantum stabilizers and preserves the exact Barnes-Wall decoding structure through a chain of isometric isomorphisms. The code distance is constant at $\Delta^2 = 1$ (in units of $2\pi$), representing an explicit distance--rate tradeoff in which logarithmic encoding efficiency is achieved at the cost of non-scaling protection. This construction provides a deterministic, space-efficient paradigm for GKP error correction in platforms supporting non-local modular connectivity.
\end{abstract}

\maketitle


Quantum error correction with continuous-variable (CV) systems~\cite{gottesman2001,terhal2016,albert2020,campagne2020} encodes logical information into the phase space of bosonic modes. The Gottesman-Kitaev-Preskill (GKP) code~\cite{gottesman2001} defines stabilizers as displacement operators corresponding to a lattice $\Lambda \subset \mathbb{R}^{2N}$. Designing practical, high-dimensional GKP block codes faces three intertwined challenges. First, existing high-rate GKP constructions generally exhibit a tradeoff between encoding rate and distance: codes with growing distance~\cite{conrad2024,bloemer2026} maintain constant rate, while high-rate constructions typically sacrifice distance. Second, dense classical lattices rarely satisfy the symplectic commutativity conditions required for valid quantum stabilizers~\cite{conrad2022}. Third, maximum-likelihood decoding on unstructured lattices requires solving the closest vector problem (CVP), which is computationally expensive for general cases~\cite{agrell2002} and conjectured hard for worst-case lattices~\cite{micciancio2001}.

Multimode GKP encoding has been pursued along several complementary directions. Noh, Girvin, and Jiang~\cite{noh2020} proposed encoding a single oscillator into many oscillators via two-mode squeezing; Royer, Singh, and Girvin~\cite{royer2022} developed multimode grid states and analyzed their stabilizer structure. Lin, Chamberland, and Noh~\cite{lin2023} studied closest-lattice-point decoding for general multimode GKP codes. Conrad, Burchards, and Flammia~\cite{conrad2024lattices} placed GKP codes in a unified lattice-theoretic framework connecting algebraic geometry and fault tolerance. Randomized constructions based on NTRU and module-SIS lattices~\cite{conrad2024,bloemer2026} achieve asymptotically good distances but rely on trapdoor or heuristic decoders whose effective decoding radii shrink with dimension. In contrast, we focus on a deterministic structured lattice---Barnes-Wall---that natively supports efficient bounded-distance decoding.


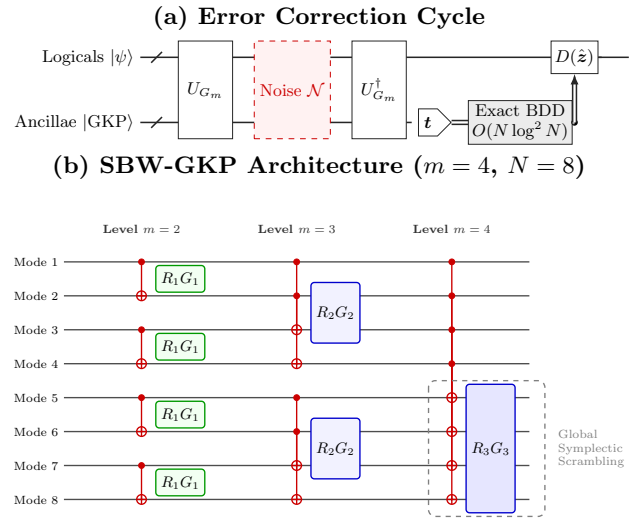
\begin{figure}[htbp]
\centering
\textbf{(a) Error Correction Cycle}\par
\vspace{0.1cm}
\resizebox{0.95\columnwidth}{!}{%
\begin{tikzpicture}[
    wire/.style={draw=black!85, thick},
    cwire/.style={draw=black!85, thick, double, double distance=1.2pt},
    tallgate/.style={draw=black!85, fill=white, rectangle, minimum width=1.0cm, minimum height=1.8cm, inner sep=3pt, font=\small},
    smallgate/.style={draw=black!85, fill=white, rectangle, minimum width=0.8cm, minimum height=0.6cm, inner sep=2pt, font=\small},
    meas/.style={draw=black!85, fill=white, signal, signal to=east, signal from=nowhere, minimum width=0.6cm, minimum height=0.5cm, inner sep=2pt}
]
\draw[wire] (0, 1.2) -- (8.5, 1.2);
\draw[wire] (0, 0) -- (5.0, 0);
\draw[thick, black!85] (0.2, 1.1) -- (0.4, 1.3);
\draw[thick, black!85] (0.2, -0.1) -- (0.4, 0.1);
\node[left, font=\small] at (0, 1.2) {Logicals $|\psi\rangle$};
\node[left, font=\small] at (0, 0) {Ancillae $|\mathrm{GKP}\rangle$};
\node[tallgate] at (1.2, 0.6) {$U_{G_m}$};
\node[tallgate, text=red!80!black, draw=red!80!black, fill=red!5, dashed] at (2.8, 0.6) {Noise $\mathcal{N}$};
\node[tallgate] at (4.4, 0.6) {$U_{G_m}^\dagger$};
\node[meas] (m) at (5.3, 0) {$\bvec{t}$};
\draw[cwire] (m.east) -- (6.2, 0);
\node[smallgate, fill=gray!15, align=center] (bdd) at (7.0, 0) {Exact BDD \\ $O(N\log^2 N)$};
\node[circle, fill=black, inner sep=1.2pt] (dot) at (8.0, 0) {};
\draw[cwire] (bdd.east) -- (dot.center);
\draw[cwire, -{Latex[length=2mm]}] (dot.center) -- (8.0, 0.9);
\node[smallgate] (disp) at (8.0, 1.2) {$D(\hat{\bvec{z}})$};
\draw[wire] (disp.east) -- (9.0, 1.2);
\end{tikzpicture}%
}
\vspace{0.4cm}
\textbf{(b) SBW-GKP Architecture ($m=4$, $N=8$)}\par
\vspace{0.1cm}
\resizebox{0.95\columnwidth}{!}{%
\begin{tikzpicture}[
    x=1.6cm, y=0.7cm,
    mode/.style={draw=black!70, thick},
    gate/.style={draw=blue!80!black, fill=blue!5, thick, rectangle, minimum width=0.8cm, minimum height=0.55cm, rounded corners=2pt, font=\small},
    ctrl/.style={draw=red!80!black, fill=red!80!black, circle, inner sep=1.2pt},
    targ/.style={draw=red!80!black, thick, circle, inner sep=2pt, path picture={\draw[thick, red!80!black] (path picture bounding box.north) -- (path picture bounding box.south) (path picture bounding box.east) -- (path picture bounding box.west);}},
    link/.style={draw=red!80!black, thick}
]
\foreach \i in {1,...,8} {
    \draw[mode] (0,-\i) -- (6,-\i);
    \node[left, font=\scriptsize] at (0,-\i) {Mode \i};
}
\node[above, font=\bfseries\scriptsize, text=black!80] at (1, -0.3) {Level $m=2$};
\foreach \i in {1,3,5,7} {
    \pgfmathtruncatemacro{\j}{\i+1}
    \draw[link] (1,-\i) -- (1,-\j);
    \node[ctrl] at (1,-\i) {}; \node[targ] at (1,-\j) {};
    \node[gate, fill=green!5, draw=green!60!black] at (1.5, -\i-0.5) {$R_1 G_1$};
}
\node[above, font=\bfseries\scriptsize, text=black!80] at (3, -0.3) {Level $m=3$};
\foreach \i in {1,2,5,6} {
    \pgfmathtruncatemacro{\j}{\i+2}
    \draw[link] (3,-\i) -- (3,-\j);
    \node[ctrl] at (3,-\i) {}; \node[targ] at (3,-\j) {};
}
\node[gate, minimum height=1.25cm] at (3.5, -2.5) {$R_2 G_2$};
\node[gate, minimum height=1.25cm] at (3.5, -6.5) {$R_2 G_2$};
\node[above, font=\bfseries\scriptsize, text=black!80] at (5, -0.3) {Level $m=4$};
\foreach \i in {1,2,3,4} {
    \pgfmathtruncatemacro{\j}{\i+4}
    \draw[link] (5,-\i) -- (5,-\j);
    \node[ctrl] at (5,-\i) {}; \node[targ] at (5,-\j) {};
}
\node[gate, minimum height=2.65cm, fill=blue!10] at (5.5, -6.5) {$R_3 G_3$};
\draw[dashed, thick, gray, rounded corners] (4.7, -4.5) rectangle (6.2, -8.5);
\node[right, font=\scriptsize, text=gray, align=left] at (6.3, -6.5) {Global\\Symplectic\\Scrambling};
\end{tikzpicture}%
}
\caption{Operational cycle and hardware architecture. (a) Logical modes and GKP ancillae are globally entangled via $U_{G_m}$; localized burst noise is scattered by $U_{G_m}^\dagger$ into diffuse background syndrome $\bvec{t}$. A deterministic $O(N\log^2 N)$ classical bounded-distance decoder computes the optimal recovery displacement $D(\hat{\bvec{z}})$. (b) The generation matrix $G_m$ operates via a hierarchical butterfly network of symplectic entangling gates ($R_m G_m$), maximally scrambling phase space to pack $R = \frac{1}{2}\log_2 N$ logical qubits.}
\label{fig:architecture_and_cycle}
\end{figure}

Consider $N$ bosonic modes with canonical operators $\bvec{\hat{r}} = (\hat{q}_1, \hat{p}_1, \dots, \hat{q}_N, \hat{p}_N)^T$ satisfying $[\hat{r}_j, \hat{r}_k] = i \symp_{jk}$, where we adopt the interleaved symplectic convention $\symp_{2N} = \bigoplus_{j=1}^N \begin{pmatrix} 0 & 1 \\ -1 & 0 \end{pmatrix}$ throughout---this choice is essential for the block-diagonal structure used in the recursive proofs below. A valid GKP stabilizer lattice $\Lambda$ with generator matrix $M$ must satisfy $M^T \symp M = 2\pi K$, where $K \in \mathbb{Z}^{2N \times 2N}$ is integral and antisymmetric~\cite{conrad2022}. The code distance is $\Delta^2 = \min_{\bvec{x} \in \Lambda^\perp \setminus \Lambda} \|\bvec{x}\|^2$, where $\Lambda^\perp$ is the symplectic dual lattice. The operational cycle [Fig.~\ref{fig:architecture_and_cycle}(a)] proceeds by globally entangling logical modes with GKP ancillae, extracting the syndrome, and applying a deterministic recovery displacement computed classically.


While classic Barnes-Wall lattices provide bounded-distance decoding structures~\cite{barnes1959,micciancio2008}, standard real generators do not natively satisfy symplectic integrality. We introduce a matrix-based block recursion over $\mathbb{R}^{2N}$ to resolve this; its hierarchical butterfly structure is realized physically as a network of symplectic entangling gates [Fig.~\ref{fig:architecture_and_cycle}(b)]. The identification with the classical Barnes-Wall lattice is established in the Supplemental Material~\cite{SM} through a unimodular equivalence over $\mathbb{Z}[i]$; throughout this Letter, the term ``Barnes-Wall lattice'' refers to the real lattice obtained from the standard Gaussian-integer Barnes-Wall lattice through this isometric identification.

We recursively define the symplectic Barnes-Wall generator. Let $m \geq 1$ index the recursion level, acting on $N = 2^{m-1}$ modes ($2^m$ phase-space dimensions):
\begin{align}
    G_1 &= I_2, \label{eq:g1} \\
    G_{m+1} &= \begin{pmatrix} G_m & 0 \\ G_m & R_m G_m \end{pmatrix}, \label{eq:grec}
\end{align}
where $R_m = I_{2^m} + \symp_{2^m}$ satisfies $R_m^T R_m = 2 I_{2^m}$. Note that $R_m/\sqrt{2}$ is simultaneously orthogonal and symplectic with unit determinant; $R_m$ itself is a $\sqrt{2}$-scaled orthogonal symplectic matrix with $\det(R_m) = 2^{2^{m-1}}$.

\begin{theorem}[Symplectic Integrality]\label{thm:main}
For any $m \geq 1$, the generator $G_m$ is invertible, and the normalized overlap matrix
\begin{equation}
K_m := G_m^T \symp_{2^m} G_m
\end{equation}
is integer-valued, antisymmetric, and non-singular. Consequently, $\tilde{M}_m = \sqrt{2\pi}\, G_m$ defines a valid multimode GKP code.
\end{theorem}

\begin{proof}
By induction. For $m=1$, $G_1 = I_2$, yielding $K_1 = \symp_2$, which is integer and antisymmetric.

Assume $K_m$ is integer and antisymmetric. Since $G_1 = I_2$ is integral and $R_m = I + \symp$ has entries in $\{0, \pm 1\}$, the recursion~\eqref{eq:grec} preserves integrality: $G_m \in M_{2^m}(\mathbb{Z})$ for all $m$. Define $S_m = G_m^T G_m$, which is symmetric and integral by the same closure argument.

Evaluating $K_{m+1}$ via block multiplication under the interleaved convention (so that $\symp_{2^{m+1}} = \symp_{2^m} \oplus \symp_{2^m}$):
\begin{align}
K_{m+1} &= \begin{pmatrix} G_m^T & G_m^T \\ 0 & G_m^T R_m^T \end{pmatrix}
           \begin{pmatrix} \symp_{2^m} & 0 \\ 0 & \symp_{2^m} \end{pmatrix}
           \begin{pmatrix} G_m & 0 \\ G_m & R_m G_m \end{pmatrix}.
\end{align}
Computing the four blocks:
\begin{align}
\text{Top-left:}&\quad G_m^T \symp G_m + G_m^T \symp G_m = 2K_m, \nonumber\\
\text{Top-right:}&\quad G_m^T \symp R_m G_m = G_m^T(\symp - I)G_m = K_m - S_m, \nonumber\\
\text{Bottom-left:}&\quad G_m^T R_m^T \symp G_m = G_m^T(\symp + I)G_m = K_m + S_m, \nonumber\\
\text{Bottom-right:}&\quad G_m^T R_m^T \symp R_m G_m = G_m^T (2\symp) G_m = 2K_m, \nonumber
\end{align}
where we used $\symp R_m = \symp(I+\symp) = \symp - I$, $R_m^T \symp = (I-\symp)\symp = \symp + I$, and $R_m^T \symp R_m = 2\symp$, all following from $\symp^2 = -I$, $\symp^T = -\symp$.

Thus:
\begin{equation}\label{eq:k-recursive}
K_{m+1} = \begin{pmatrix} 2K_m & K_m - S_m \\ K_m + S_m & 2K_m \end{pmatrix}.
\end{equation}
Antisymmetry holds because $(K_m - S_m)^T = K_m^T - S_m^T = -K_m - S_m = -(K_m + S_m)$. Integrality follows from closure under addition and multiplication over $\mathbb{Z}$. Non-singularity follows from $\det(G_{m+1}) = (\det G_m)^2 \det(R_m) \neq 0$.

By the Weyl commutation relation, $\hat{D}(\bvec{u})\hat{D}(\bvec{v}) = e^{-i \bvec{u}^T \symp \bvec{v}} \hat{D}(\bvec{v})\hat{D}(\bvec{u})$. Since every entry of $2\pi K_m = \tilde{M}_m^T \symp \tilde{M}_m$ is an integer multiple of $2\pi$, all generator displacements commute, yielding an abelian stabilizer group.
\end{proof}

\begin{corollary}[Encoding Rate and Distance]\label{cor:distance}
The logical Hilbert-space dimension equals $\mathcal{D}_m = \sqrt{\det K_m}$. Substituting the symplectic determinant formula~\cite{conrad2022} yields $\mathcal{D}_m = \det(G_m)$, which evaluates to $\mathcal{D}_m = 2^{k_m}$ with $k_m = (m-1)2^{m-2}$. The code encodes $k_m$ logical qubits into $N = 2^{m-1}$ modes, giving a logarithmic rate $R = k_m/N = \frac{1}{2}\log_2 N$. Leveraging the unimodular lattice equivalence established in the Supplemental Material~\cite{SM}, the code distance is $\Delta^2 = 1$ (in units of $2\pi$). This constant distance reflects an explicit design tradeoff: logarithmic rate growth is achieved at the cost of distance that does not scale with $N$.
\end{corollary}

Unlike finite-dimensional stabilizer codes, the logical dimension of a CV-GKP code is dictated by the phase-space volume of the lattice fundamental cell, allowing the encoding rate to exceed one logical qubit per physical mode ($R > 1$) for $N \ge 8$. This is standard for lattice GKP codes and does not contradict finite-dimensional stabilizer bounds.


To contextualize this tradeoff, we derive the geometric capacity limit for multimode GKP codes.

\begin{theorem}[Minkowski Rate--Distance Bound]\label{thm:minkowski}
For any $N$-mode GKP code with logical dimension $\mathcal{D} = 2^{RN}$, the squared distance is asymptotically bounded by
\begin{equation}
\Delta^2 \lesssim \frac{8}{e}\, N\, 2^{-R}.
\end{equation}
\end{theorem}

\begin{proof}
The GKP distance $\Delta$ is the minimum norm of nonzero vectors in the symplectic dual $\Lambda^\perp = 2\pi \symp M^{-T} \mathbb{Z}^{2N}$. Using $M^T \symp M = 2\pi K$ and $\mathcal{D} = \sqrt{\det K} = |\det M|/(2\pi)^N$, the fundamental cell volume is $V_{\mathrm{cell}}(\Lambda^\perp) = (2\pi)^N 2^{-RN}$. By the Minkowski sphere-packing principle~\cite{banaszczyk1993} and foundational GKP capacity arguments~\cite{harrington2001}, non-overlapping $2N$-dimensional Euclidean balls of radius $r = \Delta/2$ cannot exceed the cell volume:
\begin{equation}
V_{2N}(\Delta/2) = \frac{\pi^N (\Delta/2)^{2N}}{\Gamma(N+1)} \le (2\pi)^N 2^{-RN}.
\end{equation}
Canceling $\pi^N$, rearranging, and applying Stirling's approximation asymptotically (using $(N!)^{1/N} \sim N/e$) yields $\Delta^2 \le 8\,(N!)^{1/N}\,2^{-R} \approx \frac{8}{e} N 2^{-R}$.
\end{proof}

This bound reveals the severity of the distance--rate tradeoff. For constant-rate codes ($R = O(1)$), the physical limit permits $\Delta^2 \sim O(N)$. For our code with $R = \frac{1}{2}\log_2 N$, the bound restricts the theoretical maximum to $\Delta_{\max}^2 \sim O(\sqrt{N})$. Our construction deliberately accepts $O(1)$ distance---well below this ceiling---to secure the exact algebraic structure enabling deterministic $O(N \log^2 N)$ decoding, a uniquely tractable operating point compared to heuristic decoders for random lattices~\cite{bloemer2026}.

Topological approaches like the Surface-GKP code~\cite{vuillot2019,fukui2018} represent the prevailing paradigm for near-term CV fault tolerance due to their geometrically local stabilizer generators. However, enforcing 2D locality incurs severe penalties: the encoding rate vanishes as $R \sim O(1/\sqrt{N})$, and topological constraints restrict the achievable distance to $O(\sqrt{N})$ despite the $O(N)$ theoretical ceiling. Moreover, Surface-GKP relies on probabilistic graph-based decoders (e.g., minimum weight perfect matching) that face real-time computational bottlenecks for soft continuous syndromes. These limitations are illustrated in Fig.~\ref{fig:surface_gkp_limitations}: for $d=2$ a localized burst error on adjacent modes instantly forms a fatal logical string [Fig.~\ref{fig:surface_gkp_limitations}(a)], and scaling to $d=3$ requires quadratic spatial expansion while still encoding only one logical qubit [Fig.~\ref{fig:surface_gkp_limitations}(b)].

\begin{figure}[t!]
\centering
\textbf{(a) Surface-GKP $d=2$, $N=8$}\par
\vspace{0.1cm}
\resizebox{0.85\columnwidth}{!}{%
\begin{tikzpicture}[
    x=2cm, y=2cm,
    data/.style={draw=black!80, thick, circle, fill=white, minimum size=0.9cm, font=\bfseries\small},
    zsyn/.style={draw=green!50!black, thick, circle, fill=green!20, minimum size=0.9cm, font=\bfseries\small},
    xsyn/.style={draw=orange!70!black, thick, circle, fill=orange!20, minimum size=0.9cm, font=\bfseries\small},
    edge/.style={draw=black!60, very thick}
]
\fill[green!15] (0,1) -- (1,1) -- (1,0) -- (0,0) -- cycle;
\fill[orange!15] (0,1) -- (1,1) -- (0.5, 1.5) -- cycle;
\fill[orange!15] (0,0) -- (1,0) -- (0.5, -0.5) -- cycle;
\fill[green!15] (0,1) -- (0,0) -- (-0.5, 0.5) -- cycle;
\draw[edge] (0.5,0.5) -- (0,1); \draw[edge] (0.5,0.5) -- (1,1);
\draw[edge] (0.5,0.5) -- (0,0); \draw[edge] (0.5,0.5) -- (1,0);
\draw[edge] (0.5,1.5) -- (0,1); \draw[edge] (0.5,1.5) -- (1,1);
\draw[edge] (0.5,-0.5) -- (0,0); \draw[edge] (0.5,-0.5) -- (1,0);
\draw[edge] (-0.5,0.5) -- (0,1); \draw[edge] (-0.5,0.5) -- (0,0);
\node[data] (M1) at (0,1) {M1}; \node[data] (M2) at (1,1) {M2};
\node[data] (M3) at (0,0) {M3}; \node[data] (M4) at (1,0) {M4};
\node[zsyn, align=center] (M5) at (0.5,0.5) {M5\\$Z$};
\node[xsyn, align=center] (M6) at (0.5,1.5) {M6\\$X$};
\node[xsyn, align=center] (M7) at (0.5,-0.5) {M7\\$X$};
\node[zsyn, align=center] (M8) at (-0.5,0.5) {M8\\$Z$};
\node[align=left, font=\small, text=black!80] at (2.2, 1.1) {\textbf{Rate:} $R = 1/8$};
\node[align=left, font=\small, text=red!80!black] at (2.2, -0.5) {Burst errors on adjacent\\modes form fatal logical strings.};
\draw[red, ultra thick, dashed, -{Latex[length=3mm]}] (-0.4, 1.4) -- (M1.north west) node[midway, above left] {Burst Error};
\draw[red, ultra thick, dashed] (M1.south east) -- (M5.north west);
\end{tikzpicture}%
}

\vspace{0.4cm}
\textbf{(b) Surface-GKP $d=3$, $N=17$}\par
\vspace{0.1cm}
\resizebox{0.9\columnwidth}{!}{%
\begin{tikzpicture}[
    x=1.6cm, y=1.6cm,
    data/.style={draw=black!80, thick, circle, fill=white, minimum size=0.8cm, font=\bfseries\scriptsize},
    zsyn/.style={draw=green!50!black, thick, circle, fill=green!20, minimum size=0.8cm, font=\bfseries\scriptsize},
    xsyn/.style={draw=orange!70!black, thick, circle, fill=orange!20, minimum size=0.8cm, font=\bfseries\scriptsize},
    edge/.style={draw=black!60, very thick}
]
\coordinate (Q1) at (0, 2); \coordinate (Q2) at (2, 2); \coordinate (Q3) at (4, 2);
\coordinate (Q4) at (1, 1); \coordinate (Q5) at (3, 1); \coordinate (Q6) at (5, 1);
\coordinate (Q7) at (2, 0); \coordinate (Q8) at (4, 0); \coordinate (Q9) at (6, 0);
\coordinate (Z1) at (1, 2); \coordinate (Z2) at (3, 2);
\coordinate (X1) at (0.5, 1.5); \coordinate (Z3) at (2, 1); \coordinate (X2) at (4, 1);
\coordinate (X3) at (1.5, 0.5); \coordinate (Z4) at (3.5, 0.5);
\coordinate (X4) at (3, 0);
\fill[orange!15] (Q1) -- (Z1) -- (X1) -- cycle;
\fill[green!15] (Q1) -- (Z1) -- (Q4) -- (X1) -- cycle;
\fill[orange!15] (Z1) -- (Q2) -- (Z3) -- (Q4) -- cycle;
\fill[green!15] (Q2) -- (Z2) -- (Q5) -- (Z3) -- cycle;
\fill[orange!15] (Z2) -- (Q3) -- (X2) -- (Q5) -- cycle;
\fill[green!15] (X1) -- (Q4) -- (X3) -- cycle;
\fill[orange!15] (Q4) -- (Z3) -- (Q7) -- (X3) -- cycle;
\fill[green!15] (Z3) -- (Q5) -- (Z4) -- (Q7) -- cycle;
\fill[orange!15] (Q5) -- (X2) -- (Q6) -- (Z4) -- cycle;
\fill[green!15] (Q7) -- (Z4) -- (Q8) -- (X4) -- cycle;
\fill[orange!15] (Q8) -- (Z4) -- (Q9) -- cycle;
\draw[edge] (Q1) -- (Z1) -- (Q2) -- (Z2) -- (Q3);
\draw[edge] (X1) -- (Q4) -- (Z3) -- (Q5) -- (X2) -- (Q6);
\draw[edge] (X3) -- (Q7) -- (Z4) -- (Q8) -- (X4) -- (Q9);
\draw[edge] (Q1) -- (X1); \draw[edge] (Z1) -- (Q4); \draw[edge] (Q2) -- (Z3); \draw[edge] (Z2) -- (Q5); \draw[edge] (Q3) -- (X2);
\draw[edge] (X1) -- (Q4); \draw[edge] (Q4) -- (X3); \draw[edge] (Z3) -- (Q7); \draw[edge] (Q5) -- (Z4); \draw[edge] (X2) -- (Q8);
\draw[edge] (X3) -- (Q7); \draw[edge] (Q7) -- (X4); \draw[edge] (Z4) -- (Q9);
\node[zsyn] at (Z1) {$Z$}; \node[zsyn] at (Z2) {$Z$};
\node[xsyn] at (X1) {$X$}; \node[zsyn] at (Z3) {$Z$}; \node[xsyn] at (X2) {$X$};
\node[xsyn] at (X3) {$X$}; \node[zsyn] at (Z4) {$Z$}; \node[xsyn] at (X4) {$X$};
\node[data] at (Q1) {D1}; \node[data] at (Q2) {D2}; \node[data] at (Q3) {D3};
\node[data] at (Q4) {D4}; \node[data] at (Q5) {D5}; \node[data] at (Q6) {D6};
\node[data] at (Q7) {D7}; \node[data] at (Q8) {D8}; \node[data] at (Q9) {D9};
\node[align=left, font=\small, text=black!80] at (5.5, 1.8) {\textbf{Modes:} $N=17$};
\node[align=left, font=\small, text=red!80!black] at (5.5, 1.5) {\textbf{Logical:} $\mathbf{k=1}$};
\node[align=left, font=\small, text=black!80] at (5.5, 0.5) {\textbf{Rate:} $1/17 \approx 0.058$};
\end{tikzpicture}%
}
\caption{Limitations of the 2D Surface-GKP code. (a) In a $d=2$ planar layout, a localized burst error (red dashed line) on adjacent modes instantly forms a fatal logical string. (b) Scaling to $d=3$ suppresses errors but requires quadratic expansion ($N=17$) while yielding exactly one logical qubit---illustrating the vanishing rate $R \sim O(1/\sqrt{N})$ inherent to rigid geometric locality.}
\label{fig:surface_gkp_limitations}
\end{figure}
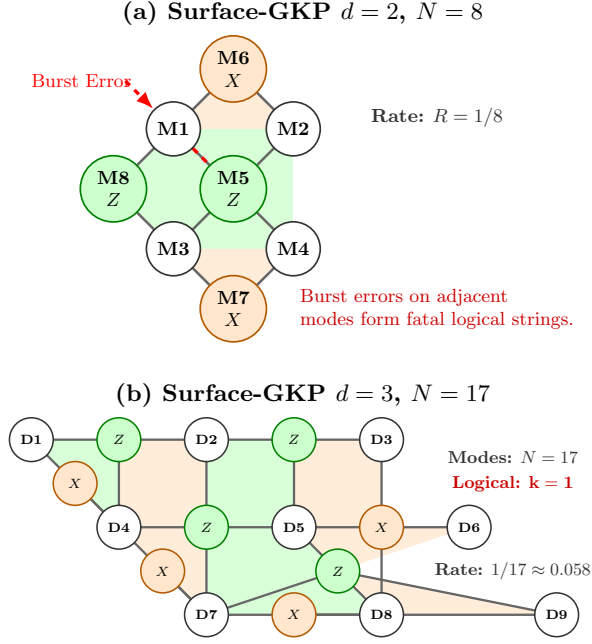

Our SBW-GKP construction abandons geometric locality for globally entangled lattice packings in high-dimensional phase space, yielding three structural advantages. First, the symplectic generator $U_{G_m}$ globally scrambles phase-space volume to exploit the extreme packing density of the Barnes-Wall lattice, achieving $R = \frac{1}{2}\log_2 N$---for instance, $BW_{16}$ packs 12 logical qubits into 8 physical modes ($R=1.5$). Second, the Voronoi decoding cell of the SBW-GKP code is a $2N$-dimensional polytope approximating a hypersphere, whereas the decoding region of $N$ independent GKP codes is a hypercube. When subjected to i.i.d.\ Gaussian displacement noise with variance $\sigma^2$ per quadrature, the total noise vector exhibits norm concentration in the $2N$-dimensional phase space, sharply peaking at $\|\bvec{e}\|^2 \approx 2N\sigma^2$. The global bounded-distance decoder evaluates this total error vector, succeeding deterministically if $\|\bvec{e}\| < \rho \sim \Delta/2 = 1/2$. Consequently, even with a constant code distance $\Delta^2 = 1$, the code guarantees exact correction provided the single-mode noise variance satisfies $\sigma^2 \lesssim 1/(8N)$. Physically, a large error spike in one mode is collaboratively absorbed into the global threshold provided surrounding modes are sufficiently quiet.  Third, the unitary structurally interleaves modes, scattering localized physical crosstalk into independent microscopic components in the dual symplectic space, which the global decoder neutralizes seamlessly.

To make this concrete, consider $N=128$ bosonic modes: a distance-$d \approx 11$ Surface-GKP code encodes one logical qubit; trivially concatenated GKP encodes 128 qubits; at recursion level $m=8$, SBW-GKP encodes $k_8 = 448$ logical qubits from the 256-dimensional geometric packing alone. Moreover, a localized burst error spanning adjacent modes, which forms a fatal logical string in surface codes, is scattered by $U_{G_m}^\dagger$ into low-amplitude background noise resolvable by the deterministic decoder.

The distinct operating points of these families are compared schematically in Fig.~\ref{fig:scaling} and summarized in Table~\ref{tab:comparison}. The rate advantage of SBW-GKP grows logarithmically with $N$ [Fig.~\ref{fig:scaling}(a)], while its bounded-distance decoding radius remains constant [Fig.~\ref{fig:scaling}(b)]---reflecting the distance--rate tradeoff inherited from the Barnes-Wall structure. Table~\ref{tab:comparison} places this tradeoff in context: random lattice constructions saturate the Minkowski capacity limit but their effective decoding radii shrink with dimension, whereas our explicit construction secures a deterministic $O(N\log^2 N)$ decoder at constant distance.


\begin{figure}[htbp]
\centering
\begin{tikzpicture}
\begin{axis}[
    name=plot1,
    width=0.48\columnwidth,
    height=4.8cm,
    xlabel={Physical Modes $N$},
    ylabel={Encoding Rate $k/N$},
    xmin=0, xmax=35,
    ymin=0, ymax=3,
    xtick={8,16,32},
    legend pos=north west,
    legend style={nodes={scale=0.6, transform shape}, draw=black!50},
    title={(a) Rate Scaling}
]
\addplot[blue, thick, mark=square*] coordinates {(2,0.5) (4,1) (8,1.5) (16,2) (32,2.5)};
\addplot[red, thick, dashed] coordinates {(2,1) (4,1) (8,1) (16,1) (32,1)};
\addplot[green!60!black, thick, mark=triangle*] coordinates {(2,0.7) (4,0.5) (8,0.35) (16,0.25) (32,0.17)};
\legend{SBW-GKP, Random M-SIS, Surface-GKP}
\end{axis}
\begin{axis}[
    name=plot2,
    at={(plot1.right of south east)}, anchor=left of south west, xshift=0.3cm,
    width=0.48\columnwidth,
    height=4.8cm,
    xlabel={Physical Modes $N$},
    ylabel={Normalized BDD Radius},
    xmin=0, xmax=35,
    ymin=0, ymax=1.2,
    xtick={8,16,32},
    legend pos=north east,
    legend style={nodes={scale=0.6, transform shape}, draw=black!50},
    title={(b) Decoding Radius}
]
\addplot[blue, thick, mark=square*] coordinates {(2,0.5) (4,0.5) (8,0.5) (16,0.5) (32,0.5)};
\addplot[red, thick, dashed, domain=2:32, samples=50] {0.4*(x/2)^(-1.25)};
\addplot[green!60!black, thick, mark=triangle*] coordinates {(2,0.5) (4,0.59) (8,0.71) (16,0.84) (32,1)};
\legend{SBW-GKP, Random M-SIS, Surface-GKP}
\end{axis}
\end{tikzpicture}
\caption{Theoretical scaling comparison of multimode GKP architectures (schematic). (a) Rate $k/N$ vs.\ $N$: SBW-GKP (blue, $\frac{1}{2}\log_2 N$) achieves logarithmic growth; Random M-SIS (red, $O(1)$) is bounded; Surface-GKP (green, $O(1/\sqrt{N})$) decays. (b) Normalized BDD radius vs.\ $N$: SBW-GKP (blue, $O(1)$) inherits a constant radius from the Barnes-Wall structure; Random M-SIS (red, $O(N^{-5/4})$) follows exact power-law decay~\cite{bloemer2026}; Surface-GKP (green, $O(\sqrt{N})$) benefits from growing surface-code distance.}
\label{fig:scaling}
\end{figure}
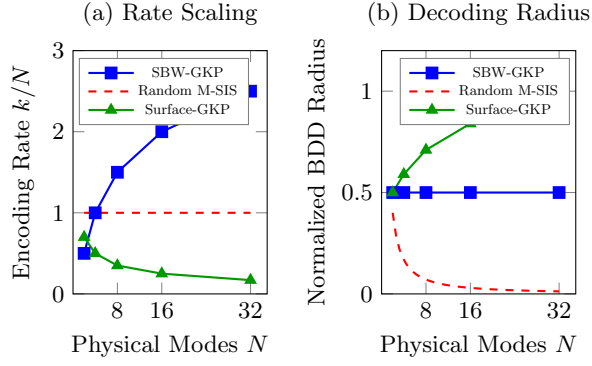

\begin{table}[t!]
\centering
\caption{Comparison of multimode GKP code families. The distance $\Delta^2$ is normalized such that $\Delta^2=1$ is the baseline GKP distance. $\Delta_{\max}^2$ denotes the asymptotic Minkowski capacity limit (Theorem~\ref{thm:minkowski}). $^\ast$For random constructions, bounded-distance decoding (BDD) restricts successful recovery to a vanishing fraction of the code distance~\cite{bloemer2026}.}
\label{tab:comparison}
\resizebox{\columnwidth}{!}{%
\begin{tabular}{@{}lcccc@{}}
\hline \hline
Code Family & \textbf{Rate $R$} & \textbf{Theoretical $\Delta_{\max}^2$} & \textbf{Actual $\Delta^2$} & \textbf{Efficient Decoding} \\
\hline
NTRU~\cite{conrad2024} & $O(1)$ & $O(N)$ & $\sim N$ & Trapdoor BDD$^\ast$ \\
M-SIS~\cite{bloemer2026} & $O(1)$ & $O(N)$ & $\sim N$ & Nearly linear$^\ast$ \\
Surface-GKP~\cite{vuillot2019} & $O(1/\sqrt{N})$ & $O(N)$ & $O(\sqrt{N})$ & $\mathrm{poly}(N)$ (MWPM) \\
\textbf{SBW-GKP (this work)} & $\mathbf{\frac{1}{2}\log_2 N}$ & $\mathbf{O(\sqrt{N})}$ & $\mathbf{O(1)}$ & \textbf{Exact BDD} $\mathbf{O(N\log^2 N)}$ \\
\hline \hline
\end{tabular}%
}
\end{table}


Our construction provides the first explicit infinite family of multimode GKP codes that simultaneously achieves $R = \frac{1}{2}\log_2 N$ and inherits an exact $O(N\log^2 N)$ bounded-distance decoder from the Barnes-Wall lattice. The decoder inheritance is rigorously established in the Supplemental Material through a chain of isometries: the complex map $\Phi$ preserves the symplectic structure, the orthogonal matrix $\Omega$ preserves Voronoi cells, and the unitary $Q$ preserves Euclidean decoding regions---yielding a decoder for the symplectic dual $\Lambda_{\mathrm{SBW}}^\perp$ with the same complexity and radius as the primal Barnes-Wall decoder.

The constant distance $\Delta^2 = 1$ reflects a deliberate distance--rate tradeoff, analogous to the capacity-rate tradeoffs in classical coding theory. For Gaussian displacement noise with variance $\sigma^2$, the BDD decoder succeeds when the error vector falls within the decoding radius $\rho \sim \Delta/2$, providing deterministic polynomial-time recovery where heuristic decoders fail. In non-asymptotic regimes ($N \leq 64$), the construction is entirely explicit---because the Barnes-Wall family provides known optimal sphere packings in low dimensions (e.g., Gosset lattice $E_8$ at $N=4$,  laminated lattice $\Lambda_{16}$ at $N=8$), the SBW-GKP code consistently achieves the upper-bound performance of these exceptional lattices without the tail risk of suboptimal random instances.

The practical viability of this approach hinges on hardware platforms supporting non-local modular connectivity, such as programmable photonic meshes with full unitary synthesis~\cite{yamasaki2020} or superconducting circuits with long-range couplers. In such architectures, the hierarchical butterfly network (Fig.~\ref{fig:architecture_and_cycle}b) maps directly to physical entangling gates, while the deterministic classical decoder eliminates the real-time scheduling bottlenecks that constrain graph-based topological decoders.


\begin{acknowledgments}
We thank J.~Bl\"omer, Y.~Xiao, Z.~Raissi, and S.~Soltan for presenting their work at ISIT~2026, and Prof.\ Hsuan-Yin Lin for introducing Ref.~\cite{conrad2024}. This work was supported by the National Key R\&D Program of China (Grant No.~2026YFE0101400) and the National Natural Science Foundation of China (Grant No.~62571218).
\end{acknowledgments}


\newpage
\clearpage
\setcounter{page}{1}

\setcounter{equation}{0}
\setcounter{figure}{0}
\setcounter{table}{0}
\renewcommand{\theequation}{S\arabic{equation}}
\renewcommand{\thefigure}{S\arabic{figure}}
\renewcommand{\thetable}{S\arabic{table}}

\section*{Supplemental Material:\\
Symplectic Barnes-Wall GKP Codes}

\subsection*{Appendix A: Unimodular Lattice Equivalence}
\label{app:isomorphism}

We prove that our real generator $G_m$ over $\mathbb{Z}$ generates the same lattice as the classical Barnes-Wall generator $B_m$ over the Gaussian integers $\mathbb{Z}[i]$. Our $B_m$ uses the column-generation convention, corresponding to the transpose of Micciancio and Nicolosi's row-generation matrix $BW^{m-1}$~\cite{micciancio2008}. Since both conventions describe the same Barnes-Wall lattice abstractly, we adopt the column convention for compatibility with our recursive structure.

Define the bijective coordinate transformation $\Phi: \mathbb{R}^{2^m} \to \mathbb{C}^{2^{m-1}}$ as $\Phi(\bvec{x})_j = q_j + i p_j$. Restricted to the integers, $\Phi(\mathbb{Z}^{2^m}) = \mathbb{Z}[i]^{2^{m-1}}$.

\begin{lemma}[Symplectic Structure Preservation]\label{lem:symp-preserve}
The map $\Phi$ intertwines the symplectic form $\symp$ with complex multiplication: for all $\bvec{x}, \bvec{y} \in \mathbb{R}^{2^m}$,
\begin{equation}
\bvec{x}^T \symp \bvec{y} = \Im\left[\Phi(\bvec{x})^\dagger \Phi(\bvec{y})\right],
\end{equation}
and consequently $\Phi(\symp \bvec{x}) = -i \Phi(\bvec{x})$.
\end{lemma}

\begin{proof}
Write $\bvec{x} = (q_1, p_1, \dots, q_N, p_N)^T$ and $\bvec{y} = (q'_1, p'_1, \dots, q'_N, p'_N)^T$, where $N = 2^{m-1}$. By definition $\Phi(\bvec{x})_j = q_j + i p_j$ and $\Phi(\bvec{y})_j = q'_j + i p'_j$. Then
\begin{align}
\Im\left[\Phi(\bvec{x})^\dagger \Phi(\bvec{y})\right]
&= \sum_{j=1}^N \Im\left[(q_j - i p_j)(q'_j + i p'_j)\right] \nonumber \\
&= \sum_{j=1}^N (q_j p'_j - p_j q'_j)
 = \bvec{x}^T \symp \bvec{y}.
\end{align}
For the second claim, note that $\symp$ acts on each mode as $\symp (q, p)^T = (p, -q)^T$, mapping to $p - i q = -i(q + i p)$ under $\Phi$. Extending mode-wise yields $\Phi(\symp \bvec{x}) = -i \Phi(\bvec{x})$.
\end{proof}

By Lemma~\ref{lem:symp-preserve}, matrix multiplication by $\symp$ maps to multiplication by $-i$. The scaled rotation $R_m = I + \symp$ maps precisely to multiplication by the scalar $\phi = 1 - i$.

Applying this transformation, the real recursion for $G_m$ translates to a complexified matrix $\mathcal{G}_m$ over $\mathbb{Z}[i]$:
\begin{equation}
\mathcal{G}_{m+1} = \begin{pmatrix} \mathcal{G}_m & 0 \\ \mathcal{G}_m & (1-i)\mathcal{G}_m \end{pmatrix}, \quad \mathcal{G}_1 = 1.
\end{equation}

The classical Barnes-Wall generator follows
\begin{equation}
B_{m+1} = \begin{pmatrix} B_m & 0 \\ B_m & (1+i)B_m \end{pmatrix}, \quad B_1 = 1.
\end{equation}

\begin{theorem}[Exact Lattice Equality]\label{thm:lattice-eq}
For all $m \geq 1$, there exists a unimodular matrix $U_m \in \mathrm{GL}_{2^{m-1}}(\mathbb{Z}[i])$ such that $\mathcal{G}_m = B_m U_m$. Consequently, $\Lambda_{\mathrm{SBW}}$ and $\Lambda_{\mathrm{BW}}$ define the same lattice over $\mathbb{Z}[i]$.
\end{theorem}

\begin{proof}
By induction. For $m=1$, $\mathcal{G}_1 = B_1 = 1$, thus $U_1 = 1$ is unimodular.

Assume $\mathcal{G}_m = B_m U_m$ for a unimodular $U_m$. We evaluate $\mathcal{G}_{m+1}$:
\begin{align}
\mathcal{G}_{m+1}
&= \begin{pmatrix} B_m U_m & 0 \\ B_m U_m & (1-i)B_m U_m \end{pmatrix} \nonumber \\
&= \begin{pmatrix} B_m & 0 \\ B_m & (1+i)B_m \end{pmatrix}
   \begin{pmatrix} U_m & 0 \\ 0 & -i U_m \end{pmatrix} \nonumber \\
&= B_{m+1} U_{m+1},
\end{align}
where we used the exact algebraic identity $1-i = -i(1+i)$. The block diagonal matrix $U_{m+1} = \mathrm{diag}(U_m, -i U_m)$ operates over the Gaussian integers. Because $-i$ is a fundamental unit in $\mathbb{Z}[i]$, the lower block $-i U_m$ remains within $M_{2^{m-1}}(\mathbb{Z}[i])$. Its determinant
\begin{equation}
\det(U_{m+1}) = (-i)^{2^{m-1}} (\det U_m)^2
\end{equation}
evaluates to a unit in $\mathbb{Z}[i]$, making $U_{m+1}$ unimodular.

Therefore, $\Lambda_{\mathrm{SBW}} = \Lambda_{\mathrm{BW}}$, establishing the exact lattice equivalence. Furthermore, since $\Phi$ is an isometry ($\|\bvec{x}\|_{\mathbb{R}}^2 = \|\Phi(\bvec{x})\|_{\mathbb{C}}^2$), the minimum norm and dual lattice spectrum are preserved.
\end{proof}

\subsection*{Appendix B: Exact BDD on the Symplectic Dual Lattice}
\label{app:bdd}

By Conway and Sloane~\cite{conway1999}, the Barnes-Wall lattice is scaled-isodual. Specifically, its Euclidean dual is given by $\Lambda_{\mathrm{BW}}^* = c Q \Lambda_{\mathrm{BW}}$, where $c = 2^{-(m-1)/2}$ and $Q$ is a unitary matrix.

For GKP codes, the relevant structure is the symplectic dual $\Lambda^\perp = \symp \Lambda^*$. Since $\symp$ is an orthogonal transformation, the symplectic and Euclidean duals are isometrically equivalent.

\begin{remark}[Code Distance and Geometric Scaling]\label{rem:distance}
The mechanism preserving the constant GKP code distance $\Delta_m^2 = 1$ can be rigorously observed through the dual scaling. For the explicit recursive generator $B_m$ (and equivalently our real $G_m$), any lattice vector takes the form
\begin{equation}
\begin{pmatrix} B_{m-1} \bvec{x} \\ B_{m-1} \bvec{x} + (1+i)B_{m-1} \bvec{y} \end{pmatrix}.
\end{equation}
Because $|1+i|^2 = 2$, the squared minimum Euclidean distance strictly doubles at each iteration, yielding $\mu(\Lambda_{\mathrm{BW}})^2 = 2^{m-1}$.

However, the GKP code distance is dictated by the shortest vector in the symplectic dual $\Lambda_{\mathrm{SBW}}^\perp$. By the isometry $\Phi$ and the scaled-isoduality $\Lambda_{\mathrm{BW}}^* = c Q \Lambda_{\mathrm{BW}}$, distances in the dual lattice are multiplied by the scaling factor $c^2 = (2^{-(m-1)/2})^2 = 2^{-(m-1)}$. Consequently, the squared GKP code distance is:
\begin{equation}
\Delta_m^2 = c^2 \mu(\Lambda_{\mathrm{BW}})^2 = 2^{-(m-1)} \times 2^{m-1} = 1.
\end{equation}
This exact geometric cancellation uniquely characterizes our construction: the recursive symplectic structure condenses the lattice packing to encode more logical qubits, while scaling the dual distance down to a perfectly invariant constant.
\end{remark}

\begin{theorem}[Decoder Inheritance]\label{thm:decoder}
Let $D_{\mathrm{BW}}$ be an $O(N \log^2 N)$-time bounded-distance decoder for the primal Barnes-Wall lattice $\Lambda_{\mathrm{BW}}$ with decoding radius $\rho$. Then the composite map
\begin{equation}\label{eq:decoder-map}
D_{\mathrm{SBW}^\perp} := \symp \circ \Phi^{-1} \circ (c Q) \circ D_{\mathrm{BW}} \circ \left(\frac{1}{c}Q^\dagger\right) \circ \Phi \circ \symp^T
\end{equation}
is a bounded-distance decoder for the symplectic dual $\Lambda_{\mathrm{SBW}}^\perp$ with the same decoding radius $\rho$, operating in $O(N \log^2 N)$ time.
\end{theorem}

\begin{proof}
For any target vector $\bvec{\tau} \in \mathbb{R}^{2N}$, decoding on $\Lambda_{\mathrm{SBW}}^\perp$ is equivalent to finding the lattice point $\hat{\bvec{\lambda}} \in \Lambda_{\mathrm{SBW}}^\perp$ that minimizes the Euclidean distance $\|\bvec{\tau} - \bvec{\lambda}\|_{\mathbb{R}}$. By sequentially applying the structural definitions $\Lambda_{\mathrm{SBW}}^\perp = \symp \Lambda_{\mathrm{SBW}}^*$, the isometry $\Phi(\Lambda_{\mathrm{SBW}}^*) = \Lambda_{\mathrm{BW}}^*$, and the isoduality $\Lambda_{\mathrm{BW}}^* = c Q \Lambda_{\mathrm{BW}}$, we have:
\begin{align}
\min_{\bvec{\lambda} \in \Lambda_{\mathrm{SBW}}^\perp} \|\bvec{\tau} - \bvec{\lambda}\|_{\mathbb{R}}
&= \min_{\bvec{u} \in \Lambda_{\mathrm{SBW}}^*} \|\bvec{\tau} - \symp \bvec{u}\|_{\mathbb{R}} \label{eq:sp1} \\
&= \min_{\bvec{u} \in \Lambda_{\mathrm{SBW}}^*} \|\symp^T \bvec{\tau} - \bvec{u}\|_{\mathbb{R}} \label{eq:sp2} \\
&= \min_{\bvec{v} \in \Lambda_{\mathrm{BW}}^*} \|\Phi(\symp^T \bvec{\tau}) - \bvec{v}\|_{\mathbb{C}} \label{eq:sp3} \\
&= \min_{\bvec{x} \in \Lambda_{\mathrm{BW}}} \|\Phi(\symp^T \bvec{\tau}) - c Q \bvec{x}\|_{\mathbb{C}} \label{eq:sp4} \\
&= c \min_{\bvec{x} \in \Lambda_{\mathrm{BW}}} \left\| \frac{1}{c}Q^\dagger \Phi(\symp^T \bvec{\tau}) - \bvec{x} \right\|_{\mathbb{C}}, \label{eq:sp5}
\end{align}
where \eqref{eq:sp2} follows from $\symp^T\symp = I$, \eqref{eq:sp3} follows from Lemma~\ref{lem:symp-preserve} wherein $\Phi$ is an isometry ($\|\bvec{a}\|_{\mathbb{R}} = \|\Phi(\bvec{a})\|_{\mathbb{C}}$), and \eqref{eq:sp5} follows from the unitarity of $Q$ ($Q^\dagger Q = I$).

The minimizing vector $\hat{\bvec{x}} \in \Lambda_{\mathrm{BW}}$ in \eqref{eq:sp5} is uniquely determined by evaluating the primal decoder $D_{\mathrm{BW}}$ at the transformed input $\tilde{\bvec{\tau}} = \frac{1}{c}Q^\dagger \Phi(\symp^T \bvec{\tau})$.

To recover the optimal lattice point $\hat{\bvec{\lambda}} \in \Lambda_{\mathrm{SBW}}^\perp$, we reverse the substitutions: $\bvec{v} = c Q \hat{\bvec{x}}$, $\bvec{u} = \Phi^{-1}(\bvec{v})$, and $\hat{\bvec{\lambda}} = \symp \bvec{u}$. Concatenating these operations yields exactly the map $D_{\mathrm{SBW}^\perp}(\bvec{\tau})$ defined in \eqref{eq:decoder-map}. Since \eqref{eq:sp1}--\eqref{eq:sp5} are strictly distance-preserving, the decoding radius remains invariant.

Algorithmically, $D_{\mathrm{SBW}^\perp}$ consists of one invocation of $D_{\mathrm{BW}}$ and six $O(N)$ matrix-vector multiplications ($\symp$, $\symp^T$, $\Phi$, $\Phi^{-1}$, $Q$, $Q^\dagger$). The asymptotic complexity is entirely bounded by the $O(N \log^2 N)$ runtime of $D_{\mathrm{BW}}$.
\end{proof}

\end{document}